\documentclass[11pt]{article}
\usepackage[letterpaper, margin=1in]{geometry}

\usepackage{graphicx}
\usepackage{amsthm,amsmath,amssymb}
\usepackage{mathabx}
\usepackage{todonotes}
\usepackage{geometry}
\usepackage{xspace}
\usepackage{xcolor}
\usepackage{thmtools}
\usepackage{algorithm}
\usepackage[noend]{algpseudocode}
\usepackage{thm-restate}
\usepackage{hyperref}
\usepackage{cleveref}

\newtheorem{definition}{Definition}[section]
\newtheorem{theorem}{Theorem}[section]

\newtheorem{lemma}[theorem]{Lemma}
\newtheorem{proposition}[theorem]{Proposition}
\newtheorem{observation}[theorem]{Obervation}
\newtheorem{claim}{Claim}[theorem]

\newcommand*{\claimproofs}{Proof of the Claim.\quad}

\newcommand{\icfd}{{\sc EF-ICFD}}

\title{Beyond Exact Fairness: Envy-Free Incomplete Connected Fair Division}

\author{
  Ajaykrishnan E S\thanks{This work was supported by NSF Grant~No.~2505099, \emph{Collaborative Research: AF Medium: Structure and Quasi-Polynomial Time Algorithms}.}\\
  University of California, Santa Barbara, USA\\
  \texttt{es.ajaykrishnan@gmail.com}
  \and
  Daniel Lokshtanov\footnotemark[1]\\
  University of California, Santa Barbara, USA\\
  \texttt{daniello@ucsb.edu}
}

\date{} 

\begin{document}
\maketitle

\begin{abstract}

We study the problem of {\sc Envy-Free Incomplete Connected Fair Division}, where exactly $p$ vertices of an undirected graph must be allocated to agents such that each agent receives a connected share and does not envy another agent's share. Focusing on agents with additive valuations, we show that the problem remains computationally hard when parameterized by $p$ and the number of agents. This result holds even for star graphs and with the input numbers given in unary representation, thereby resolving an open problem posed by Gahlawat and Zehavi (FSTTCS 2023).

In stark contrast, we show that if one is willing to tolerate even the slightest amount of envy, then the problem becomes efficient with respect to the natural parameters. Specifically, we design an Efficient Parameterized Approximation Scheme parameterized by $p$ and the number of agent types. Our algorithm works on general graphs and remains efficient even when the input numbers are provided in binary representation.
\end{abstract}

\section{Introduction}\label{sec:intro}

Fair division is a critical problem in resource allocation, focused on distributing resources among multiple agents in a way that meets certain fairness criteria. The theoretical study of this problem has been a core area of research, with work dating as far back back as the 1940s \cite{first_fair_division_steinhaus1948problem}. Fair division, beyond its theoretical appeal, has numerous practical applications, as demonstrated by platforms like Spliddit \cite{spliddit}, which use fair allocation algorithms to solve real-world problems such as sharing rent, splitting taxi fares, and assigning goods. However, for many of these practical scenarios, one ends up having to address more complex constraints. An example of this is the case where the goods are \emph{indivisible} and the bundle of items corresponding to each agent must satisfy some \emph{connectivity} requirements. The $\varphi${\sc-Connected Fair Division} ($\varphi${\sc-CFD}) problem, introduced by Bouveret et al.~\cite{cfd_bouveret2017fair} models this using a graph whose vertices represent items, where the bundles allocated to agents are required to form connected subsets of the vertices, where $\varphi$ specifies some desired notion of fairness. 

However, this model requires that every vertex of the graph be allocated to some agent, which fails to cover some natural scenarios where there is a need to be \emph{economic} in the allocation. As an example consider a museum that plans to allocate sections of a large exhibition hall for various temporary art exhibits. Each exhibit needs a connected space to display its artwork effectively, but the museum must also reserve areas for future events or special exhibits. The challenge then becomes to assign some part of the available space fairly among the groups, ensuring each receives a connected section, while leaving enough space unallocated for other purposes. Another scenario where being economic is crucial involves allocating resource-intensive goods, such as when a company rents office spaces or a startup purchases equipment for its employee teams.

Motivated by these considerations, Gahlawat and Zehavi~\cite{icfd_gahlawat2023parameterized} introduced the {\sc$\varphi$-Incomplete Connected Fair Division} problem ($\varphi${\sc-ICFD}). In this, the goal is to find a \emph{valid} and $\varphi$-allocation, that is, an allocation of exactly $p$ vertices of the graph where each agent receives a distinct connected bundle that satisfies the fairness criterion specified by $\varphi$. From the perspective of parameterized complexity, they consider this problem for various well known notions of fairness, one of them being envy-freeness. An allocation of vertices to agents is \emph{envy-free} if, for every agent, the value they assign to their own bundle is at least as large as the value they assign to any other agent's bundle. They show that, if the valuations are provided in binary representation, then the {\sc Envy-Free Incomplete Connected Fair Division} problem becomes W[1] hard parameterized by $p$ and the number of agents. However, their work leaves behind the following open problem: does {\sc Envy-Free Incomplete Connected Fair Division} admit a Fixed-Parameter Tractable (FPT) algorithm parameterized by number of agents and $p$, if the valuations are given in unary representation?

\subsection{Our Contribution}
    We resolve the open problem posed by Gahlawat and Zehavi~\cite{icfd_gahlawat2023parameterized} by proving that {\sc Envy-Free Incomplete Connected Fair Division} is W[1]-hard when parameterized by $p + |A|$, where $A$ is the set of agents, when the input numbers are in unary representation, even for highly restricted input instances such as star graphs. Our result follows via a reduction from {\sc $(k,M)$-Vector-Sum}, a multidimensional variant of the classic NP-hard problem {\sc Subset Sum}, shown to be W[1]-hard by Abboud et al.~\cite{mss_w1_abboud2014losing}, parameterized by $k + d$, where $d$ is the dimension of the input vectors, even when the input numbers are in unary representation.

    \begin{restatable}{theorem}{Hardness}\label{thm:w1_hardness}
        {\sc Envy-Free Incomplete Connected Fair Division} is W[1]-hard parameterized by $p\,+\,|A|$, with the input numbers provided in unary representation, even on star graphs.
    \end{restatable}
    
    To address the intractability implied by Theorem~\ref{thm:w1_hardness}, we approach the problem from the perspective of approximate fairness. The study of finding \emph{approximately fair} allocations, meaning allocations that are at most a small multiplicative (or additive factor) unfair to the agents, was initiated by Lipton et al.~\cite{approx_cake_lipton2004approximately} in the context of cake-cutting and has since been investigated in other settings. For a comprehensive overview of this topic, we refer the reader to the survey by Amanatidis et al.~\cite{survey_amanatidis2023fair}, which discusses approximate fairness in various settings. In our case, an allocation is \emph{$\varepsilon$-envy-free} if, for every agent, the value they assign to their own bundle scaled up by a multiplicative factor of $(1+\varepsilon)$, is at least as large as the value they assign to any other agent's bundle. 
    
    In this setting, it is very natural to apply a geometric rounding scheme to round the involved numbers to the nearest power of $(1+\varepsilon)$ and make some guesses regarding the agents valuation of allocated items. This is especially tempting given that the hardness of the instance in the proof of Theorem~\ref{thm:w1_hardness} arises precisely from the inability to iterate over all possible valuations that agents may assign to the allocated items. This does in fact lead to a straightforward FPT algorithm, parameterized by $p + |A|$, that finds a valid and $\varepsilon$-envy-free allocation or correctly concludes that the input instance has no valid and envy-free allocation, if the input numbers are provided in unary representation. But, in the case where the input numbers are provided in binary representation, the agents could have wildly different valuations for the same items and consequently the number of guesses will no longer be FPT in terms of the relevant parameters. In fact, due to this reason, for some time we believed that there exists a positive real number $c$ such that, assuming the Exponential Time Hypothesis~\cite{IMPAGLIAZZO2001367}, no FPT algorithm can either find a valid and $c$-envy-free allocation or correctly conclude that the input instance admits no valid and envy-free allocation. Surprisingly, this is not the case -- the problem admits an algorithm with the aforementioned guarantees that runs in FPT time parameterized by $p$ and $|\mathcal{A}|$, the number of \emph{agent types}, which is a smaller parameter than the number of agents.

    \begin{restatable}{theorem}{EPAS}\label{thm:epas_param_by_p_and_|A|}
        There exists an algorithm that, given an instance $(G, A, \mathcal{U}, p)$ of {\sc Envy-Free Incomplete Connected Fair Division} with input numbers in binary representation, and a positive real number $\varepsilon$, runs in time $(\frac{1}{\varepsilon}\log\frac{p}{\epsilon})^{\mathcal{O}(|\mathcal{A}|p^2)}(mn)^{\mathcal{O}(1)}$ and outputs either a valid and $\varepsilon$-envy-free allocation $\Pi$, or correctly states that $(G, A, \mathcal{U}, p)$ does not admit a valid and envy-free allocation.
    \end{restatable}

    We now give a sketch of the algorithm asserted in Theorem~\ref{thm:epas_param_by_p_and_|A|}. The algorithm begins by applying a simple reduction rule that allows us to upper bound the number of agents by a function of $p$ and $|\mathcal{A}|$. After this, we proceed to apply the standard technique of color coding~\cite{color_coding_10.1145/210332.210337} (see also~\cite{parameterized_algorithms_10.5555/2815661}) to color the vertices of the graph with $p$ colors, and make some guesses regarding the color of the vertices that end up being assigned to each agent. 
    
    At this point, we draw inspiration from the polynomial-time algorithm of Gan et al.~\cite{house_allocation_gan2019envy}, which addresses the special case where each agent is allocated exactly one item, sometimes referred to as the \emph{house allocation} problem. To illustrate, consider a more restricted setting in which every agent is assigned a unique color and must be allocated exactly one vertex which has the same color. In this setting, for each agent, we consider the set of items they value the most. If there exists an agent such that none of their most-valued items have the same color as them, then these items can be safely discarded, as allocating them to any agent would make this agent envious. On the other hand, if for every agent there exists at least one most-valued item that has the same color as them, then assigning each agent such an item yields an envy-free allocation. The key idea here is that we can either find an envy-free allocation, or restrict the search-space by a significant amount. 
    
    Our setting is more technically challenging since we need to work with approximate envy-freeness, enforce connectivity constraints, and accommodate agents that receive multiple items. For this reason, we encode the restriction of the search space using a \emph{target} vector $\bar{\mu} = (\mu_1,\dots,\mu_{|A|})$, and an allocation \emph{attains} the target for agent $i$, if the valuation that agent $i$ has for their own bundle is between $(1+\varepsilon)^{\mu_i-1}$ and $(1+\varepsilon)^{\mu_i}$. The algorithm guesses a set of properties (of which there are only FPT many) of the valid and envy-free allocation we are searching for, sets the target vector so that $(1+\varepsilon)^{\mu_i} \geq pM_i$, where $M_i$ denotes the valuation of agent~$i$’s favorite item, and runs a subroutine that searches for a valid and envy-free allocation that respects these guesses. 
    
    The subroutine has an outer loop which, following the key idea from the simpler setting, either finds a solution or decreases the target in each iteration. To prove correctness, we must show that if there exists a valid and envy-free allocation that respects the guesses and attains a target that is at most the current one, then the algorithm does one of the following: it either returns a valid and $\varepsilon$-envy-free allocation, or identifies an agent $i$ such that no valid and envy-free allocation that respects the guesses and attains a target that is at most the current one attains the current target for agent $i$. We design a procedure that achieves this by observing that the guesses constrain the instances to the point that the allocations corresponding to each agent can be searched for independently of each other. Finding such a feasible assignment for each individual agent is done by reducing to a weighted version of the well studied {\sc Graph Motif} problem, first introduced by Lacroix et al.~\cite{lacroix2006motif}. This concludes our sketch.

    We remark that a subtle but crucial aspect for the correctness of the algorithm is that the feature guessing must be performed outside the subroutine. Additionally, we draw the reader’s attention to the fact that we consider two subtly different versions of the problem. In the optional setting, agents may be assigned an empty bundle, whereas in the mandatory setting, each agent must receive a non-empty set of items. Gahlawat and Zehavi~\cite{icfd_gahlawat2023parameterized} study only the mandatory setting. Our results apply to both the optional and mandatory settings. We first establish our results for the optional setting and then show how they extend to the mandatory case.
 
\subsection{Related Work}
    Fair division of indivisible goods have recieved a considerable amount of attention in the computational social choice literature, an interested reader can refer some of the recent surveys~\cite{survey_amanatidis2023fair,survey_suksompong2021constraints}. The work introducing the fair division of \emph{indivisible} items with \emph{connectivity} requirements on the players' bundles was carried out by Bouveret et al.~\cite{cfd_bouveret2017fair}. Their initial results inspired a series of follow-up studies on the model. This body of work includes studies on conditions guaranteeing the existence of fair allocations~\cite{almost_ef_bilo2022almost}, achieving various fairness notions such as Pareto-optimal allocations~\cite{pareto_optimal_igarashi2019pareto} and maximin share allocations~\cite{unary_greco2020complexity}. Additionally, Aziz et al.~\cite{chores_aziz2022fair} examined scenarios where the items to be allocated could be either goods or bads. 
    There has also been work, examining the \emph{price of connectivity}~\cite{price_of_connectivity_bei2022price}, which evaluates the loss in fairness when allocation are required to be connected, when compared to an optimal disconnected allocation. The study of connected fair division where only a subset of the items are to be allocated, was initiated by Gahlawat and Zehavi~\cite{icfd_gahlawat2023parameterized}. Related notions that have been explored in literature include \emph{donating} a few items to achieve fairness~\cite{charity_approx_prop_ef_caragiannis2022little} and imposing \emph{cardinality constraints} on the number of items of different types that an agent can receive~\cite{cardinality_constraints_biswas2018fair}.

    The concept of finding \emph{approximately fair} allocations that are at most a small multiplicative (or additive factor) unfair to the agents was introduced by Lipton et al.~\cite{approx_cake_lipton2004approximately} in the concept of cake-cutting and later on by Budish~\cite{10.1145/1807406.1807480} in the context of indivisible goods. Since then, it has been studied across a variety of settings and fairness notions, including envy-freeness~\cite{plaut2020almost}, Nash social welfare~\cite{Akrami_Chaudhury_Hoefer_Mehlhorn_Schmalhofer_Shahkarami_Varricchio_Vermande_Wijland_2022}, and max–min fairness~\cite{10.1145/3391403.3399526, feige2021tight}, among others. 

    $u_i(v_c)$

    The area of parameterized complexity has found notable use in fair division literature, as many fundamental problems in this domain are (unsurprisingly) NP-hard. The works most relevant to ours include that of Deligkas et al.~\cite{param_cfd_deligkas2021parameterized}, which provides a comprehensive overview of the parameterized complexity of {\sc$\varphi$-CFD} with respect to the natural parameters as well as several structural parameters of the graph for various fairness notions, and that of Gahlawat and Zehavi~\cite{icfd_gahlawat2023parameterized} which does the same for for the {\sc $\varphi$-ICFD} problem. There is a large body of work separate from this~\cite{nguyen2023fair, param_ef_bliem2016complexity}, which consider fair division problems from a parameterized perspective with the added requirement of Pareto-optimality. The works of Bredereck et al.~\cite{bredereck2019high,bredereck2023high} apply $N$-fold Integer Programming to solve problems concerning envy-free and Pareto-efficient allocations of items, with the latter even handling instances where input numbers are given in binary representation. Furthermore, the problem of finding \emph{locally} envy-free allocations, where the input includes a graph in which vertices represent agents, and an agent can only envy their neighbors in the graph have also been explored~\cite{local_ef_beynier2019local, local_ef_bredereck2022envy}. 

\section{Preliminaries}\label{sec:prelims}
For a positive integer $n$, let $[n]$ be the set $\{1,2,\dots,n\}$, and $[a,b]$ denote the set of integers $i\in\mathbb{Z}$ that satisfies $i\leq b$ and $i\geq a$. Let $\mathbb{N}_{\geq 0} = \mathbb{N}\cup\{0\}$. For vectors $\bar{v}^{(1)}, \bar{v}^{(2)} \in \mathbb{Z}^n$, we write $\bar{v}^{(1)} \leq \bar{v}^{(2)}$ (respectively, $\bar{v}^{(1)} = \bar{v}^{(2)}$) to mean that $v^{(1)}_i \leq v^{(2)}_i$ (respectively, $v^{(1)}_i =v^{(2)}_i$), where $v^{(1)}_i$ and $v^{(2)}_i$ denote the $i$-th coordinates of $\bar{v}^{(1)}$ and $\bar{v}^{(2)}$, for every $i \in [n]$. Furthermore, we use $\bar{v}^{(1)} + \bar{v}^{(2)}$ to denote the vector $\Bar{v} \in \mathbb{Z}^n$ that satisfies $v_i = v^{(1)}_i + v^{(2)}_i$ for every $i\in [n]$. We consider finite, simple, undirected graphs. For standard graph terminology we refer to the book by Diestel \cite{Diestel_10.5555/3134208}. We make a slight departure from standard notations by denoting the number of vertices in a graph by $m$.

%
%

An input instance of the {\sc Envy-Free Incomplete Connected Fair Division} problem consists of a tuple $(G,A,\mathcal{U},p)$, where,\\[-8pt]
\begin{itemize}
    \item $G$, the \emph{utility graph}, is a graph on $m$ vertices
    \item $p \in [0, m]$ is a non-negative integer
    \item $A = [n]$ is the set of \emph{agents}, where $n\in\mathbb{N}$
    \item $\mathcal{U}$ is an $n$-tuple of \emph{valuation functions} $\{ u_i:V(G)\rightarrow\mathbb{N}_{\geq 0}\}_{i\in[n]}$
\end{itemize}

The vertices of $G$ are referred to as \emph{items}. As is standard in literature, we assume that agents have \emph{additive valuations} and use $u_i(S)$ to denote $\sum_{v\in S}u_i(v)$ for any $S\subseteq V(G)$.
An \emph{allocation} of items is a tuple $\Pi = (\pi_1,\dots,\pi_n)$, where $\pi_i \subseteq V(G)$ is referred to as the \emph{bundle} of agent $i$. When clear from the context, we sometimes abuse notations and use $\pi$ to denote $\bigcup_{i\in A}\pi_i$, the set of allocated vertices. An allocation is \emph{valid} if it allocates exactly $p$ vertices in total, with each bundle being connected and no vertex allocated to multiple bundles. An allocation is \emph{envy-free} if no agent desires another agent's bundle more than their own. 
Formally, an allocation that is both valid and envy-free must satisfy,\\[-8pt]
\begin{itemize}
    \item $\pi_i\cap\pi_j = \emptyset$ for every pair of distinct agents $i,j$
    \item $\sum_{i=1}^{n}\lvert\pi_i\rvert = p$
    \item $G[\pi_i]$ is connected
    \item $u_i(\pi_i)\geq u_i(\pi_j)$ for every pair of agents $i,j$
\end{itemize}

Similarly, an allocation is \emph{$\varepsilon$-envy-free}, for some positive real number $\varepsilon$, if it satisfies $(1+\varepsilon)u_i(\pi_i)\geq u_i(\pi_j)$ for every agent $i$ and $j$. Given an instance $\mathcal{I} = (G, A, \mathcal{U}, p)$, the {\sc Envy-Free Incomplete Connected Fair Division} problem asks whether there exists an allocation in $\mathcal{I}$ that is both valid and envy-free. 
Here we remark that there are two subtly different versions of the problem. In the \emph{mandatory} setting, each agent must be assigned a non-empty bundle; that is, $|\pi_i| > 0$ for every agent $i \in A$. In contrast, the \emph{optional} setting allows agents to receive empty bundles, meaning $|\pi_i| \geq 0$ for all $i \in A$. Throughout this paper, we use the shorthand \icfd\ to refer to the envy-free incomplete connected fair division problem in the optional setting.

Two agents $i,j$ have the same \emph{type} if $u_i(v) = u_j(v)$ for every $v\in V(G)$. We use $\mathcal{A}_{\mathcal{I}}$ to denote the set of all agent types in an instance $\mathcal{I}$ of \icfd. Thus $\mathcal{A}_{\mathcal{I}}$ is a partition of $A$ such that two agents $i$, $j$ belong to the same part $\boldsymbol{a}\in\mathcal{A}_{\mathcal{I}}$ if and only if they have the same type. For ease of notation, we shall drop the subscript $\mathcal{I}$ whenever it is clear from the context. We say that an instance $(G, A, \mathcal{U}, p)$ of \icfd\ is \emph{in unary} (respectively, \emph{in binary}) if the valuation functions in $\mathcal{U}$ are encoded in unary (respectively, binary). This means there exists a polynomial (respectively, exponential) function $f$ such that $\max\{ u_i(v)\; |\; i\in A,\, v\in V(G)\} \leq f(|V(G)|)$.

For standard terminology in parameterized algorithms, we refer to the book Parameterized Algorithms \cite{parameterized_algorithms_10.5555/2815661} by Cygan et al. An \emph{$(n, k)$-perfect hash family} \cite{splitter_10.5555/795662.796315} is a family of functions from $[n]$ to $[k]$ such that for every $S \subseteq [n]$ with $|S| = k$, there exists a function in the family that is injective on $S$. The following is known regarding $(n, k)$-perfect hash families.

\begin{proposition}[\cite{splitter_10.5555/795662.796315}, \cite{parameterized_algorithms_10.5555/2815661}]\label{prop:n_k_perfect_hash_family}
    For any $n,k \in \mathbb{N}$, one can construct an $(n, k)$-perfect hash family of size $e^k k^{\mathcal{O}(\log k)}\log n$ in time $e^k k^{\mathcal{O}(\log k)}n\log n$.
\end{proposition}

An \emph{Efficient Parameterized Approximation Scheme} (EPAS) for an optimization problem $\Pi$ with input size $n$ and parameter $k$ is an algorithm that, for every $\varepsilon > 0$, runs in time $f(k, 1/\varepsilon)\cdot n^{\mathcal{O}(1)}$ for some computable function $f$, and returns a $(1+\varepsilon)$-approximate solution (or $(1-\varepsilon)$-approximate solution for maximization problems) to $\Pi$.

\section{W[1] Hardness of \icfd\ Parameterized by p and  $\lvert A\rvert$}\label{sec:hardness}
To show W[1]-hardness of \icfd\ parameterized by $p$ and $|A|$, we reduce from {\sc $(k,M)$-Vector-Sum}, a multidimensional version of the classic NP-hard problem Subset Sum. We begin by formally defining the problem and stating a known result regarding its complexity.

\begin{definition}
    For positive integers $k$, $n$, $M$, and $d$, the \text{{\sc $(k,M)$-Vector-Sum}} problem is defined as follows: given a tuple $(\mathcal{W}, \bar{t}, k)$, where $\mathcal{W} = \{w^{(i)}\}_{i\in [n]} \subseteq [0, kM]^d$, and $\Bar{t}\in [0, kM]^d$, determine whether there exists a subset $S \subseteq [n]$ of cardinality $k$, such that $\sum_{i \in S} w^{(i)} = \bar{t}$.
\end{definition}

\begin{proposition}[Lemma 4.1 - \cite{mss_w1_abboud2014losing}]\label{prop:MSS_w1_hardness}
    \text{{\sc $(k,M)$-Vector-Sum}} is W[1]-hard parameterized by $k+d$ even when $M = k\,n^{1+o(1)}$.
\end{proposition}

With the definition of $(k, M)$-\textsc{Vector-Sum} in place, we are now prepared to describe the reduction.

\subsection{The Reduction}\label{subsec:reduction}

Let $\mathcal{I}_{\mathrm{vs}} = (\mathcal{W}, \bar{t}, k)$ be an instance of $(k,n^{1+o(1)})$-\textsc{Vector-Sum} with $d$ dimensions. We define the following instance $\mathcal{I}_{\scriptscriptstyle \mathrm{FD}} = (G,A,\mathcal{U},p)$  of \icfd:

\begin{itemize}
    \item Let $G$ be a star with center $c$, and $2d+n$ leaves. We assign a unique label to each leaf using elements of the set $V_{\alpha} \cup V_{\beta} \cup \{v_1, v_2, \dots, v_n\}$, where $V_{\alpha} = \{a_1, a_2, \dots, a_d\}$ and $V_{\beta} = \{b_1, b_2, \dots, b_d\}$.  With a slight abuse of notation, we use the label to refer to the corresponding vertex.

    \item Let $A = [2d + 1]$. We assign a unique label to each agent using elements of the set $\{A_1, A_2, \dots, A_d\} \cup \{B_1, B_2, \dots, B_d\} \cup \{C\}$. Here as well, we slightly abuse notation, and use the label to refer to the corresponding agent.
    
    
    \item $p = 2d+k+1$. 
    
    \item For every $l\in[d]$, we define $u_{A_l}(x) =  
        \begin{cases} 
            n^3+w^{(l)}_i & x = v_i\\
            n^3 & x = c \\
            0 & x\in V_{\beta}  \\
            0 & x\in V_{\alpha}\setminus\{a_l\}\\
            (k+1)n^3+t_l & x = a_l\\
        \end{cases}$
        
    \item For every $l\in[d]$, we define $u_{B_l}(x) =  
        \begin{cases} 
            n^3-w^{(l)}_i & x = v_i\\
            n^3 & x = c \\
            0 & x\in V_{\alpha}  \\
            0 & x\in V_{\beta}\setminus\{b_l\}\\
            (k+1)n^3-t_l & x = b_l \\
        \end{cases}$
    \item $u_C(v) = 1$\quad for every $v\in V(G)$
\end{itemize}


\subsection{Proof of Correctness}

We shall prove that $(G,p,A,\mathcal{U})$ has a valid and envy-free allocation if and only if $(\mathcal{W},\bar{t},k)$ has a subset $S\subseteq [n]$ of cardinality $k$, such that $\sum_{i\in S}w^{(i)} = \Bar{t}$. To this end, we make some observations about the properties that any valid and envy-free allocation in the instance $\mathcal{I}_{\scriptscriptstyle \mathrm{FD}}$ must satisfy. Note that since the graph is a star and since the bundles must be connected, only the agent who is allocated the center of the star gets more than one vertex. Furthermore, every agent must receive at least one vertex. Otherwise, as each agent has zero valuation for at most $2d-1$ vertices, and since a total of $2d+k+1$ vertices must be allocated, any agent with an empty bundle would be envious of some other agent. Hence, for exactly $2d+k+1$ vertices to be assigned in total, one agent must receive $k+1$ and the others must receive one vertex each. Now, since $C$ values every vertex equally, they must receive exactly $k+1$ vertices; otherwise, they would envy the agent who has more than one vertex. With this, we are ready to prove a key property of this instance.

\begin{lemma}\label{lem:heavy_vertices}
    In any valid and envy-free allocation $\Pi$ in $(G,p,A,\mathcal{U})$, every agent except $C$ must be allocated the vertex they value the most.
\end{lemma}
\begin{proof}
    Let us consider agent $C$. Observe that, if $\pi_C$ contains a vertex $a_l$ for some $l\in [d]$, then it cannot be assigned to $A_l$, since $\Pi$ is valid. Therefore, irrespective of which vertex is assigned to $A_l$, we have $u_{A_l}(\pi_C) \geq (k+1)n^3 > u_{A_l}(\pi_{A_l})$, which contradicts the assumption that $\Pi$ is envy-free. Hence $\pi_C$ must not contain any vertex from the set $V_{\alpha}$ and by a similar argument from $V_{\beta}$. Now observe that, according to any agent $X$ distinct from $C$, we have $u_X(\pi_C)\geq (k+1)n^3 - k\,n^{1+o(1)} \geq kn^3$. Therefore, for $X$ to not be envious of $C$, we must have $u_X(\pi_X) \geq kn^3$. But this can only be achieved if $X$ is assigned the vertex they value the most. Hence, in the allocation $\Pi$, agent $A_l$ must receive $a_l$ and agent $B_l$ must receive $b_l$ for every $l \in [d]$. This concludes the proof of the lemma.
\end{proof}

 This gives us all the tools we need to prove the main lemma.
 
\begin{lemma}\label{lem:MSS_iff_ICFD}
   The instance $(G, p, A, \mathcal{U})$ admits a valid and envy-free allocation $\Pi$ if and only if there exists a subset $S \subseteq [n]$ of cardinality $k$ such that $\sum_{i \in S} w^{(i)} = \bar{t}$ in the instance $(\mathcal{W}, \bar{t}, k)$.

\end{lemma}
\begin{proof}
    Let $\Pi$ be a valid and envy-free allocation in $(G,p,A,\mathcal{U})$. According to lemma \ref{lem:heavy_vertices}, in $\Pi$, agent $A_l$ and $B_l$ must be allocated $a_l$ and $b_l$ respectively, for every $l\in[d]$. Furthermore, we also have that agent $C$ must be allocated exactly $k+1$ vertices. Let these vertices be $\{c,v_{i_1},v_{i_2},\dots,v_{i_k}\}$. We claim that the set $S = \{i_1, i_2, \dots, i_k\} \subseteq [n]$ satisfies $\sum_{i \in S} w^{(i)} = \bar{t}$ in the instance $(\mathcal{W}, \bar{t}, k)$. Observe that, as $\Pi$ is envy-free, $u_{A_l}(\pi_C)$ is at most $u_{A_l}(\pi_{A_l})$ and $u_{B_l}(\pi_C)$ is at most $u_{B_l}(\pi_{B_l})$. This implies that,
    $$n^3+\sum_{j=1}^{k}(n^3+w_l^{(i_j)}) \quad \leq \quad (k+1)n^3+t_l \quad \implies \quad \sum_{j=1}^{k}w_l^{(i_j)}\quad \leq \quad t_l$$
    $$n^3+\sum_{j=1}^{k}(n^3-w_l^{(i_j)}) \quad \leq \quad (k+1)n^3-t_l \quad \implies \quad \sum_{j=1}^{k}w_l^{(i_j)}\quad \geq \quad t_l$$
    Hence if $(G,p,A,\mathcal{U})$ has a valid and envy-free allocation, then there exists a set $S = \{i_1, i_2, \dots, i_k\} \subseteq [n]$ that satisfies $\sum_{i \in S} w^{(i)} = \bar{t}$ in the instance $(\mathcal{W}, \bar{t}, k)$.

    Let $S = \{i_1, i_2, \dots, i_k\} \subseteq [n]$ be a set of cardinality $k$, such that, $\sum_{i \in S} w^{(i)} = \bar{t}$ in the instance $(\mathcal{W}, \bar{t}, k)$. Consider an allocation $\Pi$ which allocates $\{c,v_{i_1},v_{i_2},\dots,v_{i_k}\}$ to $C$, $a_l$ to $A_l$, and $b_l$ to $B_l$ for every $l\in[d]$. We claim that $\Pi$ is a valid and envy-free allocation in $(G,p,A,\mathcal{U})$. Observe that, for every agent $X \neq C$, we have, $u_C(\pi_X)\, =\, 1\, <\, (k+1)\, =\, u_C(\pi_C)$. Furthermore, for every $l \in [d]$ and agent $X\neq A_l$, we have $u_{A_l}(\pi_X)\, \leq\, (k+1)n^3 + \sum_{j=1}^{k} w_l^{(i_j)}\, =\, (k+1)n^3 + t_l\, =\, u_{A_l}(\pi_{A_l})$. Finally, for every $l \in [d]$ and agent $X\neq B_l$, we have $u_{B_l}(\pi_X)\, \leq\, (k+1)n^3 - \sum_{j=1}^{k} w_l^{(i_j)}\, =\, (k+1)n^3 - t_l\, =\, u_{B_l}(\pi_{B_l})$.
    Consequently, $\Pi$ is envy-free for the instance $(G, p, A, \mathcal{U})$. Since $\Pi$ allocates a total of $p$ vertices as disjoint, connected subsets to the agents by definition, its validity follows. This completes the proof of the lemma.
\end{proof}

By combining Proposition~\ref{prop:MSS_w1_hardness} and Lemma~\ref{lem:MSS_iff_ICFD} with the properties of the reduction, we obtain Theorem~\ref{thm:w1_hardness}

\Hardness*
\begin{proof}
    Given an instance $\mathcal{I}_{\mathrm{vs}} = (\mathcal{W}, \bar{t}, k)$ of $(k,n^{1+o(1)})$-\textsc{Vector-Sum} with $d$ dimensions, let $\mathcal{I}_{\scriptscriptstyle \mathrm{FD}} = (G,A,\mathcal{U},p)$ be an instance of \icfd\ constructed according to Subsection~\ref{subsec:reduction}. Observe that in our construction $G$ is a star graph, $p=2d+k+1$ and $|A| = 2d+1$. Furthermore, the range of all valuation functions are bounded above by $\mathcal{O}(n^3)$ which implies that the instance is in unary and that $|\mathcal{I}_{\scriptscriptstyle \mathrm{FD}}| = n^{\mathcal{O}(1)}$. Finally, given $\mathcal{I}_{\mathrm{vs}}$, we can construct $\mathcal{I}_{\scriptscriptstyle \mathrm{FD}}$ in $n^{\mathcal{O}(1)}$ time. Therefore, the statement of the theorem follows from Proposition~\ref{prop:MSS_w1_hardness} and Lemma~\ref{lem:MSS_iff_ICFD}.
\end{proof}

We conclude by remarking that this hardness proof extends to the mandatory setting without modification, as the properties of the constructed instance ensure that, in any envy-free allocation, every agent is assigned a non-empty set. Consequently, both directions of the proof remain valid even in this context.

\section{Approximation Scheme for \icfd\ Parameterized by p and $|\mathcal{A}|$}\label{sec:algo}

We begin this section by stating and proving the reduction rule that helps us upper bound the number of agents in an instance of \icfd\ by a function of $p$ and $|\mathcal{A}|$.

\begin{lemma}\label{lem:reduction_rule}
    There exists an algorithm that takes as input an instance $\mathcal{I} = (G, A, \mathcal{U}, p)$ of \icfd, runs in $(mn)^{\mathcal{O}(1)}$ time, and outputs an instance $\mathcal{I}' = (G', A', \mathcal{U}', p')$ of \icfd\ such that $\mathcal{I}'$ admits a valid and envy-free allocation if and only if $\mathcal{I}$ admits one, and furthermore, the number of agents of type $\boldsymbol{a}$ in $\mathcal{I}'$ is at most $p' + 1$, for every $\boldsymbol{a}\in \mathcal{A}_{\mathcal{I}'}$. 
\end{lemma}
\begin{proof}
    We define $G' := G$, $p' := p$, and $\mathcal{U}' := \mathcal{U}$. To define $A'$, consider each agent type $\boldsymbol{a} \in \mathcal{A}_{\mathcal{I}}$. If there are at most $p + 1$ agents of type $\boldsymbol{a}$ in $A$, we include all of them in $A'$. Otherwise, we include an arbitrary selection of $p + 1$ agents of type $\boldsymbol{a}$ in $A'$. It is clear from definition of $\mathcal{I}'$ that the number of agents of type $\boldsymbol{a}$ in $\mathcal{I}'$ is at most $p' + 1$, for every $\boldsymbol{a}\in \mathcal{A}_{\mathcal{I}'}$. Hence to prove the lemma, it suffice to prove that $\mathcal{I}'$ admits a valid and envy-free allocation if and only if $\mathcal{I}$ admits one.

    To this end, let $\Pi'$ be a valid and envy-free allocation in $\mathcal{I}'$. We define an allocation $\Pi$ in $\mathcal{I}$, where for agent $i\in A$, we have $\pi_i := \pi'_i$ if $i\in A'\cap A$ and $\pi_i := \emptyset$ otherwise. Since the validity of $\Pi$ follows directly from that of $\Pi'$ we are left with proving envy-freeness. Note that no agent is envious of agents in $A\setminus A'$, since they are allocated the empty bundle. Furthermore, agents in $A\cap A'$ are not envious of each other since $\Pi'$ is an envy-free allocation. Finally, agents of $A\setminus A'$ are not envious of agents in $A\cap A'$, since corresponding to an agent $i\in A\setminus A'$, there exists an agent $j\in A\cap A'$ who has the same type and is not envious of anyone in $A\cap A'$. Therefore, we conclude that $\Pi$ is a valid and envy-free allocation in $\mathcal{I}$.

    Now, let $\Pi$ be a valid and envy-free allocation in $\mathcal{I}$. Let $\sigma: A' \rightarrow A$ be an injective function with the property that $\sigma(i)$ and $i$ are of the same type, and every agent in $A$ who is allocated a non-empty bundle under $\Pi$ lies in the image of $\sigma$. 
    Such a function exists because at most $p$ agents receive non-empty bundles under $\Pi$, and $A'$ includes all agents from types with at most $p+1$ agents, as well as at least $p+1$ agents from types that has more than $p+1$ agents in $A$.
    We define an allocation $\Pi'$ in $\mathcal{I}'$ by setting $\pi'_i := \pi_{\sigma(i)}$ for each $i \in A'$. Since $\Pi$ is a valid and envy-free allocation in $\mathcal{I}$, and $\Pi'$ is obtained by relabeling agents via a injection that preserves types, it follows that $\Pi'$ is also valid and envy-free in $\mathcal{I}'$. This concludes the proof of the lemma.
\end{proof}

\begin{observation}\label{obs:bound_on_agents}
    If $\mathcal{I} = (G, A, \mathcal{U}, p)$ is an instance of \icfd\ obtained after applying the algorithm of Lemma~\ref{lem:reduction_rule}, then $|A| \leq (p+1) \cdot |\mathcal{A}|$
\end{observation}
\begin{observation}\label{obs:ef_for_non_reduced_instance}
    If $\mathcal{I}'$ is the instance of \icfd\ obtained after applying the algorithm of Lemma~\ref{lem:reduction_rule} to $\mathcal{I}$, then given a valid and envy-free allocation $\Pi'$ for $\mathcal{I}'$, we can construct one for $\mathcal{I}$ in $(mn)^{\mathcal{O}(1)}$ time.
\end{observation}

Observation~\ref{obs:bound_on_agents} follows from the fact that each agent is associated with a type in $\mathcal{A}$, and the lemma guarantees that no type in the output instance has more than $p + 1$ agents. Observation~\ref{obs:ef_for_non_reduced_instance} follows from the construction described in the proof of Lemma~\ref{lem:reduction_rule}, which shows how to obtain a valid and envy-free allocation for the original instance from one for the output instance.\\[-6pt]

Now, we state a lemma that serves as a subroutine in the algorithm. Since it is not central to the main discussion, we defer its proof to the end of the section.

\begin{restatable}{lemma}{GraphMotifAlgo}\label{lem:graph_motif_algo}
There exists an algorithm that, given a graph $G$, a weight function $w: V(G) \rightarrow \mathbb{N}_{\geq 0}$, and a partition $\mathcal{C} = \{C_1, \dots, C_k\}$ of $V(G)$ into $k$ parts, runs in time $3^k \cdot |V(G)|^{\mathcal{O}(1)}$ and outputs a connected subgraph $H$ of $G$ that maximizes $\sum_{v \in V(H)} w(v)$ subject to the constraint that $|V(H) \cap C_i| = 1$ for every $i \in [k]$, if such a subgraph exists, and correctly reports that no such subgraph exists otherwise.
\end{restatable}

Now, let $\mathcal{I} = (G, A, \mathcal{U}, p)$ be an instance of \icfd\ and $\varepsilon$ be a positive real number. We make the following definitions corresponding to $\mathcal{I}$ and $\epsilon$. Let $\varepsilon_{\varepsilon}' > 0$ denote the positive root of the equation $x^2+ 3x = \varepsilon$, and define $q_{\varepsilon} := 1 + \varepsilon'_{\varepsilon}$ and $t_{\varepsilon} := 1 + \left\lceil \log_{q_{\varepsilon}}\left(\frac{p}{\varepsilon'_{\varepsilon}}\right) \right\rceil$. Finally, for positive real number $\rho$ and $\alpha\in\mathbb{N}_{\geq 0}$, let $\overline{\log}_{\rho}\alpha$ denote $\left\lceil \log \alpha \right\rceil$, if $\alpha>0$ and $-1$ otherwise.

\begin{definition}
    We define a \emph{bucket profile} to be a tuple $(\mathcal{C}, \mathcal{S}, \Lambda)_{\varepsilon,\mathcal{I}}$, where $\mathcal{C} = \{C_c\}_{c \in [p]}$ is a partition of $V(G)$, $\mathcal{S} = \{S_i\}_{i \in A}$ is a partition of $[p]$, and $\Lambda : A \times [p] \rightarrow [0, t_{\varepsilon}]$ is a function such that, for every pair of distinct agents $i, j \in A$, we have $\sum_{\substack{c \in S_j \\ \Lambda(i, c) \neq t_{\varepsilon}}} q_{\varepsilon}^{-\Lambda(i, c)} \leq q_{\varepsilon}$.
\end{definition}


Referring back to the high-level sketch of the algorithm in Section~\ref{sec:intro}, we note that the set $\mathcal{C}$ corresponds to the partitioning obtained through color coding, $\mathcal{S}$ captures the guess that maps colors to agents, and $\Lambda$ encodes the guessed features of the allocation. In particular, $\Lambda(i,c)$ represents the guess on how far apart (in terms of powers of $q$) the valuation of a vertex in color class $c$ and the valuation of the bundle allocated to agent $i$ is, with respect to $u_i$, in the valid and envy-free allocation we are searching for. The following definition formalizes what it means for an allocation to be consistent with the guesses.

\begin{definition}
    An allocation $\Pi$ in $\mathcal{I}$ \emph{respects} a bucket profile $(\mathcal{C}, \mathcal{S}, \Lambda)_{\varepsilon,\mathcal{I}}$ if it satisfies the following:
    \begin{itemize}
        \item For every agent $i \in A$ and $c \in [p]$, we have $|\pi_i \cap C_c| = 1$ if $c \in S_i$, and $|\pi_i \cap C_c| = 0$ otherwise;
        \item For every agent $i \in A$, $c \in [p]$, and $v \in \pi \cap C_c$, we have 
        \begin{itemize}
            \item $\overline{\log}_{q_{\varepsilon}}(u_i(\pi_i)) - \overline{\log}_{q_{\varepsilon}}(u_i(v)) \geq \Lambda(i, c)$ or $u_i(v) = 0$ if $\Lambda(i, c) = t_{\varepsilon}$, and,
            \item $ \overline{\log}_{q_{\varepsilon}}(u_i(\pi_i)) - \overline{\log}_{q_{\varepsilon}}(u_i(v)) = \Lambda(i, c)$ and $u_i(v) > 0$ otherwise.
        \end{itemize}
    \end{itemize}
\end{definition}

Finally, we formally define the target vector and what it means for an allocation to be consistent with it.

\begin{definition}
    Let $M := \max\{u_i(v)\; |\; v\in V(G),\, i\in A\}$. We define a \emph{target vector} $\bar{\mu}_{\varepsilon,\mathcal{I}} \in \left[-1, \lceil \log_{q_{\varepsilon}}(p M) \rceil \right]^A$ and use $\mu_i$ to denote the entry of $\bar{\mu}_{\varepsilon,\mathcal{I}}$ associated with agent $i \in A$. An allocation $\Pi$ \emph{attains} the target $\bar{\mu}_{\varepsilon,\mathcal{I}}$ for agent $i\in A$, if $\overline{\log}_{q_{\varepsilon}} (u_i(\pi_i)) = \mu_i$.
\end{definition}

For the rest of the discussion, for ease of notation, we drop the subscript $\mathcal{I}$ and $\varepsilon$, whenever it is clear from context. \\[-6pt]

The following lemma is the main subroutine of the algorithm.

\begin{lemma}\label{lem:subroutine_house_allocation}
    There exists an algorithm that, given an instance $(G, A, \mathcal{U}, p)$ of \icfd, a positive real number $\varepsilon$, a bucket profile $(\mathcal{C}, \mathcal{S}, \Lambda)$, and a target vector $\bar{\mu}$, runs in time $3^p \cdot m^{\mathcal{O}(1)}$ and returns one of the following:
    \begin{itemize}
        \item A valid and $\varepsilon$-envy-free allocation $\Pi$, or
        \item An agent $i \in A$ such that no valid and envy-free allocation that respects the bucket profile $(\mathcal{C}, \mathcal{S}, \Lambda)$ attains a target vector $\bar{\mu}'$, where $\bar{\mu}' \leq \bar{\mu}$ with $\mu_i' = \mu_i$.
    \end{itemize}
\end{lemma}
\begin{proof}
    We define vertex $v \in V(G)$ to be \emph{eligible} if, letting $c \in [p]$ be such that $v \in C_c$, for every agent $i \in A$ with $c \notin S_i$, we have $\overline{\log}_q(u_i(v)) \leq \mu_i - \Lambda(i, c)$ or $u_i(v) = 0$. Let $G_i$ be the subgraph of $G$ induced by the eligible vertices in $\bigcup_{c \in S_i} C_c$. We define a subgraph $H_i$ of $G_i$ to be \emph{colorful}, if $|V(H_i)\cap C_c| = 1$ for every $c \in S_i$. Before describing the algorithm, we prove the following claims:

    \begin{claim}\label{claim:property_one_of_subroutine}
        If for every agent $i \in A$ there exists a connected, colorful subgraph $H_i$ of $G_i$ satisfying the inequality $\overline{\log}_q \left( u_i(V(H_i)) \right) \geq \mu_i$, then the allocation $\Pi$ defined by $\pi(i) := V(H_i)$ is a valid and $\varepsilon$-envy-free allocation.
    \end{claim}
    \begin{proof}
        We begin by observing that $\{\pi_i\}_{i\in A}$ are pairwise disjoint. Otherwise, if there exists a vertex $v \in \pi_i \cap \pi_j$ for some pair of distinct agents $i, j \in A$, then $v\in V(G_i)\cap V(G_j)$ which implies that either $v$ belongs to two distinct sets in $\mathcal{C}$, which contradicts the fact that $\mathcal{C}$ is a partition of $V(G)$, or the set $C_c$ containing $v$ belongs to both $S_i$ and $S_j$, which contradicts the fact that $\mathcal{S}$ is a partition of $[p]$. Moreover, we have $\sum_{i \in A} |\pi_i| = \sum_{i \in A} |S_i| = p$, since $\mathcal{S}$ is a partition of $[p]$. Finally, as each $G[\pi_i]$ is connected by definition, it follows that $\Pi$ is a valid allocation.
        
        Now, to prove $\varepsilon$-envy-freeness, consider $c \in [p]$, agent $i \in A$ such that $c\notin S_i$, and vertex $v \in \pi \cap C_c$. If $u_i(v)$ is positive, then $\overline{\log}_q \left( u_i(v) \right) = \left\lceil \log_q \left( u_i(v) \right) \right\rceil \geq 0$. Moreover, since $\Lambda(i, c)$ is non-negative and $v$ is eligible, it follows that $0\, \leq\, \mu_i\, \leq\, \overline{\log}_q \left( u_i(\pi_i) \right) = \left\lceil \log_q \left( u_i(\pi_i) \right) \right\rceil$. Therefore, whenever $u_i(v)$ is positive, we have the following inequality, :
        \begin{align}\label{ineq:one}
           \hspace{63pt} u_i(v) \cdot q^{\Lambda(i, c)} \quad
            \nonumber & \leq \quad u_i(v) \cdot q^{\mu_i - \overline{\log}_q(u_i(v))} \\
            \nonumber & \leq \quad u_i(v) \cdot q^{\log_q(u_i(\pi_i)) + 1 - \log_q(u_i(v))} \\
            & = \quad u_i(\pi_i) \cdot q 
        \end{align}
        Thus, if $i, j \in A$ are distinct agents, then,
        \begin{align*}
            u_i(\pi_j) \quad = \quad \sum_{\substack{c\in S_j\\ v\in \pi_j\cap C_c}} u_i(v) \quad
            & \leq \quad \sum_{\substack{c\in S_j\\ \Lambda(i,c)= t}} u_i(\pi_i)\cdot q^{1-\Lambda(i,c)} + \sum_{\substack{c\in S_j\\ \Lambda(i,c)\neq t}} u_i(\pi_i)\cdot q^{1-\Lambda(i,c)}\\
            & \leq \quad p \cdot u_i(\pi_i)\cdot q^{1-t}\ +\ q^2 u_i(\pi_i) \\
            & \leq \quad \varepsilon'u_i(\pi_i)\ +\ (1+\varepsilon')^2 u_i(\pi_i) \\
            & \leq \quad (1+\varepsilon)u_i(\pi_i)
        \end{align*}
        where the first transition follows from the fact that $\pi_j$ contains exactly one vertex from each $C_c$ for $c \in S_j$; the second follows from inequality~\eqref{ineq:one}; the third uses the observation that there are at most $p$ elements $c \in S_j$ with $\Lambda(i, c) = t$ and the assumption that $(\mathcal{C}, \mathcal{S}, \Lambda)$ is a bucket profile, which implies $\sum_{\substack{c \in S_j \\ \Lambda(i,c) \neq t}} q^{-\Lambda(i,c)} \leq q$; the fourth and fifth follow from the definitions of $q$, $t$, and $\varepsilon'$. Therefore, $\Pi$ is $\varepsilon$-envy-free, which completes the proof of the claim.
    \end{proof}

    \begin{claim}\label{claim:property_two_of_subroutine}
        If there exists an agent $i \in A$ such that no connected, colorful subgraph $H_i$ of $G_i$ satisfies $\overline{\log}_q \left( u_i(V(H_i)) \right) \geq \mu_i$, then there does not exist a valid and envy-free allocation that respects $(\mathcal{C}, \mathcal{S}, \Lambda)$ and attains a target vector $\bar{\mu}'$, where $\bar{\mu}'$ is a target vector that satisfies $\bar{\mu}' \leq \bar{\mu}$ and $\mu'_i = \mu_i$.
    \end{claim}
    \begin{proof}
        Let $i \in A$ be an agent such that no connected, colorful subgraph $H_i$ of $G_i$ satisfies $\overline{\log}_q \left( u_i(V(H_i)) \right) \geq \mu_i$. Let $\Bar{\mu}'$ be a target vector that satisfies $\Bar{\mu}'\leq \Bar{\mu}$ and $\mu'_i = \mu_i$. Suppose for contradiction that there exists a valid and envy-free allocation $\Pi$ that respects $(\mathcal{C}, \mathcal{S}, \Lambda)$, and attains $\Bar{\mu}'$. 

        Since $\Pi$ is valid, $G[\pi_i]$ is connected. Moreover, as $\Pi$ attains $\bar{\mu}'$, we have $\overline{\log}_q(u_i(\pi_i)) = \mu_i' = \mu_i$. Furthermore, observe that since $\Pi$ respects $(\mathcal{C}, \mathcal{S}, \Lambda)$, each vertex in $\pi_i$ belongs to a distinct set in $\mathcal{C}$ whose indices belong to $S_i$, which ensures that $G[\pi_i]$ is colorful. Finally, to show that $G[\pi_i]$ is a subgraph of $G_i$, note that for every $c \in S_i$, $v \in C_c$, and agent $j \in A$ that is distinct from $i$, if $u_j(v) > 0$ then we have $\overline{\log}_q(u_j(\pi_j)) - \overline{\log}_q(u_j(v)) \geq \Lambda(j, c)$. As $\Pi$ attains $\bar{\mu}'$, it follows that $\overline{\log}_q(u_j(\pi_j)) = \mu'_j \leq \mu_j$, which implies $\overline{\log}_q(u_j(v)) \leq \mu_j - \Lambda(j, c)$. Hence, every vertex in $\pi_i$ is eligible, and consequently, $G[\pi_i]$ is a subgraph of $G_i$. Hence, $G[\pi_i]$ is a connected, colorful subgraph of $G_i$ satisfying $\overline{\log}_q \left( u_i(V(H_i)) \right) \geq \mu_i$, which contradicts the definition of $i$. This completes the proof of the claim.
    \end{proof}

    \textit{Algorithm:} For each agent $i \in A$, the algorithm computes the graph $G_i$, defined to be the subgraph of $G$ induced by the eligible vertices in $\bigcup_{c \in S_i} C_c$. It calls the algorithm of Lemma~\ref{lem:graph_motif_algo} with input $(G_i, u_i, \{C_c\}_{c \in S_i})$ as a subroutine to compute a maximum-weight connected, colorful subgraph $H_i$ of $G_i$. If for every agent $i \in A$ a subgraph $H_i$ is found, for which $\overline{\log}_q \left( u_i(V(H_i)) \right) \geq \mu_i$, then the algorithm returns the allocation $\Pi$ defined by $\pi(i) := V(H_i)$. Otherwise, it returns an agent $i \in A$ for which no connected, colorful subgraph satisfying $\overline{\log}_q \left( u_i(V(H_i)) \right) \geq \mu_i$ could be found.\\[-8pt]
    
    The correctness of the algorithm follows from Claims~\ref{claim:property_one_of_subroutine} and~\ref{claim:property_two_of_subroutine}, along with the correctness of the subroutine. The running time bound follows from the fact that the subroutine takes $3^p m^{\mathcal{O}(1)}$ time, along with the observation that all other steps of the algorithm run in $m^{\mathcal{O}(1)}$ time. This completes the proof of Lemma~\ref{lem:subroutine_house_allocation}.
\end{proof}

With this we are ready to prove Theorem~\ref{thm:epas_param_by_p_and_|A|}.

\EPAS*
\begin{proof}
    If the instance $(G, A, \mathcal{U}, p)$ satisfies $|A| > (p+1)|\mathcal{A}|$, then we apply the algorithm from Lemma~\ref{lem:reduction_rule} to obtain a new instance $\mathcal{I}'$, for which we have $|A_{\mathcal{I}'}| \leq (p+1)|\mathcal{A}_{\mathcal{I}'}|$, by Observation~\ref{obs:bound_on_agents}. We then solve $\mathcal{I}'$ and, using Observation~\ref{obs:ef_for_non_reduced_instance} and Lemma~\ref{lem:reduction_rule}, reconstruct a solution for the original instance $(G, A, \mathcal{U}, p)$, incurring only a negligible overhead of $m^{\mathcal{O}(1)}$ time. Therefore, for the remainder of the proof, we may assume that we are working with an instance that satisfies $|A| \leq (p + 1)|\mathcal{A}|$.
    
    Now, let us fix an arbitrary labeling $l$ on $V(G)$ using the set $[m]$. We say that a partition $\mathcal{C}$ of $V(G)$ is induced by a function $f: [m] \rightarrow [p]$, if for every $v\in V(G)$, we have that $v\in C_c$ if and only if $f(l(v)) = c$, where $l(v)$ denotes the label of $v$. Let $\mathcal{B}$ denote the set of all bucket profiles that belong to the set $\mathcal{F}_{\mathcal{C}} \times \mathcal{F}_{\mathcal{S}} \times \mathcal{F}_{\Lambda}$, where $\mathcal{F}_{\mathcal{C}}$ is the set of partitions of $V(G)$ induced by elements of an $(m, p)$-perfect hash family of size $e^p p^{\mathcal{O}(\log p)}\log m$ whose existence is guaranteed by Proposition~\ref{prop:n_k_perfect_hash_family}, $\mathcal{F}_{\mathcal{S}}$ is the set of all partitions of the set $[p]$ into $n$ parts, where some parts may be empty, and $\mathcal{F}_{\Lambda}$ is the set of all functions from $A \times [p]$ to $[0,t]$. Let $M := \max\{u_i(v) \mid v \in V(G),\, i \in A\}$, and define a target vector $\bar{\mu}' = (\mu'_1, \dots, \mu'_n)$, where $\mu'_i := \lceil \log_q(p \cdot M) \rceil$ for each $i \in A$. \\[-10pt]
    
    \textit{Algorithm:} For each bucket profile in $\mathcal{B}$, the algorithm initializes the target vector $\bar{\mu}$ to $\bar{\mu}'$. It then invokes the algorithm from Lemma~\ref{lem:subroutine_house_allocation} as a subroutine. If it returns a valid and $\varepsilon$-envy-free allocation $\Pi$, the algorithm outputs $\Pi$ and terminates. Otherwise, upon receiving an agent $i \in A$, it updates $\bar{\mu}$ by reducing $\mu_i$ by one and iterates by invoking the algorithm from Lemma~\ref{lem:subroutine_house_allocation} again. This loop continues until some entry of $\bar{\mu}$ becomes equal to negative two, at which point the algorithm proceeds to the next bucket profile in $\mathcal{B}$. If no allocation is found after all elements of $\mathcal{B}$ have been considered, the algorithm concludes that the instance $(G, A, \mathcal{U}, p)$ admits no valid and envy-free allocation.\\[-10pt]

    To prove correctness of the algorithm, assume that $(G,A,\mathcal{U},p)$ is an instance of \icfd\ that has a valid and envy-free allocation $\Pi$. We only need to show that the algorithm returns an allocation, since by Lemma~\ref{lem:subroutine_house_allocation}, any allocation that is returned by the subroutine and consequently the algorithm will be valid and $\varepsilon$-envy-free.
    Toward this end, we make the following claim.
    
    \begin{claim}\label{claim:if_allocation_then_bucket_profile}
        If $(G,A,\mathcal{U},p)$ is an instance of \icfd\ that has a valid and envy-free allocation $\Pi$, then there exists a bucket profile $(\mathcal{C}, \mathcal{S}, \Lambda) \in \mathcal{B}$ that is respected by $\Pi$.
    \end{claim}
    \begin{proof}
        Let $\mathcal{C}$ be an element of $\mathcal{F}_{\mathcal{C}}$, such that for every $c \in [p]$, we have that $|C_c\cap \bigcup_{i \in A} \pi_i| = 1$. To see that such an element exists, observe that since $\Pi$ is a valid allocation, the sets $\pi_i$ are pairwise disjoint and $\sum_{i \in A} |\pi_i| = p$, which implies that $|\bigcup_{i \in A} \pi_i| = p$. Thus, the set of labels assigned to vertices in $\bigcup_{i \in A} \pi_i$ forms a subset of $[m]$ of size $p$. By the definition of an $(m, p)$-perfect hash family, there exists a function that is injective on this set, and this function induces a partition of $V(G)$ satisfying the desired property. Let $\mathcal{S}$ be a partition of $[p]$ such that $c\in [p]$ belongs to $S_i$ if and only if $C_c\cap \pi_i \neq \emptyset$. Finally, for every agent $i \in A$, $c \in [p]$ and $v\in \pi \cap C_c$, define, $\Lambda(i, c) = \overline{\log}_q(u_i(\pi_i)) - \overline{\log}_q(u_i(v_c))$ if it is less than $t$ and $u_i(v_c) > 0$, and $\Lambda(i, c) = t$ otherwise. The tuple $(\mathcal{C}, \mathcal{S}, \Lambda)$ is a bucket profile, since for any pair of distinct agents $i,j\in A$, we have, : 
        $$
        \sum_{\substack{c\in S_j\\\Lambda(i,c)\neq t}} q^{-\Lambda(i,c)} \quad
        \leq \quad \sum_{\substack{c\in S_j\\ v\in\pi\cap C_c}} q^{-\left\lceil \log_q(u_i(\pi_i)) \right\rceil + \left\lceil \log_q(u_i(v_c)) \right\rceil} \quad
        \leq \quad q \cdot \sum_{v\in\pi_j} \frac{u_i(v)}{u_i(\pi_i)} \quad
        \leq \quad q
        $$
        Finally, since $\Pi$ respects the bucket profile $(\mathcal{C}, \mathcal{S}, \Lambda)$ by definition, the lemma follows.
    \end{proof}
    
    Let $(\mathcal{C}, \mathcal{S}, \Lambda)$ be the bucket profile in $\mathcal{B}$ that is respected by $\Pi$, whose existence is guaranteed by Claim~\ref{claim:if_allocation_then_bucket_profile}. Let $\bar{\mu}^o$ denote the target vector attained by $\Pi$, where $\mu^o_i := \lceil \log_q(u_i(\pi_i)) \rceil$ for each $i \in A$. 
    We claim that, while the algorithm processes the bucket profile $(\mathcal{C}, \mathcal{S}, \Lambda)$, it never considers any target vector $\bar{\mu}$ for which $\mu_i < \mu^o_i$ for some $i \in A$. Suppose, for contradiction, that it does. Let $\bar{\mu}^{\alpha}$ be the first such target vector considered by the algorithm and $\Bar{\mu}^\beta$ be the target vector of the previous iteration. Let $i\in A$ be the agent returned by the subroutine while considering $(\mathcal{C}, \mathcal{S}, \Lambda)$ and $\Bar{\mu}^\beta$. By definition of $\bar{\mu}^\alpha$, we have that $\mu^\alpha_i = \mu^o_i - 1$ and $\mu^\alpha_j \geq \mu^o_j$ for every $j\in A$ such that $j\neq i$, which implies $\Bar{\mu}^o\leq \Bar{\mu}^\beta$ and $\mu^{\beta}_i = \mu^o_i$. 
    But, since $\Pi$ is a valid and envy-free allocation that respects $(\mathcal{C}, \mathcal{S}, \Lambda)$ and attains $\bar{\mu}^o$ for agent $i$, by the correctness of Lemma~\ref{lem:subroutine_house_allocation}, agent $i$ cannot be returned by the subroutine in the iteration corresponding to $(\mathcal{C}, \mathcal{S}, \Lambda)$ and target vector $\bar{\mu}^\beta$. Thus, we conclude that while the algorithm processes the bucket profile $(\mathcal{C}, \mathcal{S}, \Lambda)$, it never considers any target vector $\bar{\mu}$ for which $\mu_i < \mu^o_i$ for some $i \in A$.
    Observe that this suffices to prove that the algorithm returns an allocation, since the algorithm begins the iteration corresponding to $(\mathcal{C}, \mathcal{S}, \Lambda)$ with $\bar{\mu} = \bar{\mu}'$, and in every iteration where it does not return an allocation, it decreases $\mu_i$ for some $i \in A$. Therefore, if the algorithm never encounters a target vector $\bar{\mu}$ with $\mu_i < \mu_i^o$ for any $i \in A$, it must eventually consider $\bar{\mu} = \bar{\mu}^o$, before reaching a target vector with at least one entry equal to negative two. By the correctness of the subroutine in Lemma~\ref{lem:subroutine_house_allocation}, an allocation must be returned at this point, since $\Pi$ is a valid and envy-free allocation that respects $(\mathcal{C}, \mathcal{S}, \Lambda)$ and attains $\Bar{\mu}^o$ for every agent $i$. This concludes the proof of correctness of the algorithm. 

    Finally, the running time bound of $(\frac{1}{\varepsilon}\log\frac{p}{\varepsilon})^{\mathcal{O}(|\mathcal{A}|p^2)}m^{\mathcal{O}(1)}$ follows from the fact that the subroutine runs in $3^p m^{\mathcal{O}(1)}$ time, together with the observation that the algorithm considers at most $e^p p^{\mathcal{O}(\log p)}\log m \cdot (p |\mathcal{A}|)^p \cdot t^{|\mathcal{A}|p^2}$ bucket profiles in $\mathcal{B}$, and that the entries of $\bar{\mu}$ can be decreased at most $\mathcal{O}(p\, m\, |\mathcal{A}|)$ times before at least one of them becomes negative.
\end{proof}

We remark that this algorithm extends to the mandatory setting with only minor modification. In particular, it suffices to strengthen the definition of a bucket profile to ensure that the partition $\mathcal{S}$ only contains non-empty sets. The correctness of the algorithm then follows with straightforward changes to account for this additional constraint.

Finally, we conclude the section by restating Lemma~\ref{lem:graph_motif_algo} and describing the algorithm it refers to.

\GraphMotifAlgo*
\begin{proof}
    Let $H$ be a subgraph of $G$. We use the term \emph{weight of $H$} to mean the quantity $\sum_{v\in V(H)}w(v)$. We describe a dynamic programming algorithm that computes the weight of the maximum-weight connected subgraph $H$ of $G$ that satisfies $|V(H)\cap C_i| = 1$ for every $i\in [k]$, if such a subgraph exists. Using standard backtracking techniques or self-reduction, a subgraph that realizes this value can be recovered with, in the worst case, an additional overhead of $n^{\mathcal{O}(1)}$ in the running time.

    First, let us describe a function $f : V(G)\times 2^{[k]} \rightarrow \mathbb{N}_{\geq 0}$, where $2^{[k]}$ denotes the set of all subsets of $[k]$.
    \begin{align}\label{recurrance}
        f(v, S) = 
        \begin{cases}
            w(v),  &\quad \text{if } S = \{i\} \text{ and } v \in C_i, \\
            -\infty, &\quad \text{if } v \notin \bigcup_{i \in S} C_i, \\
            \displaystyle\max_{\substack{\emptyset \subsetneq S_1 \subsetneq S \\ u \in N[v]}} \left\{ f(u, S_1) + f(v, S \setminus S_1) \right\}, & \quad\text{otherwise}.
        \end{cases}
    \end{align}
    
    It is straightforward to observe that $f(v, S)$ computes the weight of the maximum-weight connected subgraph $H$ of $G[\bigcup_{i \in S} C_i]$ that contains the vertex $v$. The algorithm constructs a dynamic programming table $F[v, S]$, where $v \in V(G)$ and $S \subseteq [k]$, by filling entries in increasing order of $|S|$ starting with singleton sets. The entries corresponding to larger sets, are computed using the recurrence of function $f$ in Definition~\eqref{recurrance}, after replacing recursive calls to $f$ with lookups to the table $F$ at the corresponding indices. After populating the table states, the algorithm checks the value of $\max_{v\in V(G)}F[v,[k]]$. If it is $-\infty$ the algorithm reports that no such subgraph exists, otherwise it returns the value. The correctness of the algorithm follows from an induction on the size of the subset $S$ in the definition of the table entries. Finally, we observe that the running time is upper bounded by 
    $$
    \sum_{\substack{v\in V(G)\\ S\subseteq [p]}}2^{|S|}\cdot |V(G)| \
    \leq \sum_{v\in V(G)}\sum_{i\in [k]} \binom{k}{i}\cdot 2^i\cdot 1^{k-i}\cdot |V(G)| \
    = \sum_{v\in V(G)}3^k\cdot |V(G)| \
    = \ 3^k |V(G)|^{2},
    $$ 
    which concludes the proof of the lemma.
\end{proof} 

\section{Concluding Remarks and Discussion}
We prove that \icfd is W[1]-hard parameterized by $p+|A|$, even when the input numbers are in unary representation and provide an EPAS parameterized by $p+|\mathcal{A}|$ for the problem even when input is in binary representation. We remark that our algorithm outperforms the brute force algorithm running in $\mathcal{O}(m^p\cdot n^p)$ when $p < \mathcal{O}(\log m + \log n)$. This regime remains practically relevant, particularly in scenarios involving the allocation of expensive or limited goods.

Beyond envy-freeness, other fairness notions have been explored in the incomplete and connected setting. Gahlawat and Zehavi~\cite{icfd_gahlawat2023parameterized} investigated three such notions, including EF1 and EFX, which relax envy-freeness by requiring that envy disappears once the most (respectively, least) valuable item is removed from every other agent’s bundle. They showed that finding an EF1 (EFX) allocation parameterized by $p + |A|$ is W[1]-hard when the input is given in binary representation. Our results may provide a starting point for answering related questions to these relaxed fairness concepts.

\bibliography{bib2doi}

\end{document}